\documentclass[aps,preprint]{revtex4}%
\usepackage{amsfonts}
\usepackage{amsmath}
\usepackage{amssymb}
\usepackage{graphicx}%
\providecommand{\U}[1]{\protect\rule{.1in}{.1in}}
\newtheorem{theorem}{Theorem}

\newtheorem{remark}[theorem]{Remark}

\newenvironment{proof}[1][Proof]{\noindent\textbf{#1.} }{\ \rule{0.5em}{0.5em}}
\begin{document}
\preprint{ }
\title[ ]{Quaternionic structures, supertwistors and fundamental superspaces}
\author{Diego Julio Cirilo-Lombardo}
\affiliation{National Institute of Plasma Physics (INFIP), Consejo Nacional de
Investigaciones Cientificas y Tecnicas (CONICET), Facultad de Ciencias Exactas
y Naturales, Universidad de Buenos Aires, Cuidad Universitaria, Buenos Aires
1428, Argentina}
\affiliation{Bogoliubov Laboratory of Theoretical Physics, Joint Institute for Nuclear
Research, 141980 Dubna, Russia}
\email{diego777jcl@gmail.com}
\author{Victor N. Pervushin}
\affiliation{Bogoliubov Laboratory of Theoretical Physics, Joint Institute for Nuclear
Research, 141980 Dubna, Russia}
\author{}
\affiliation{Other Institution}
\pacs{PACS number}

\begin{abstract}
Superspace is considered as space of parameters of the supercoherent states
defining the basis for oscillator-like unitary irreducible representations of
the generalized superconformal group $SU(2m,2n\mid2N)$ in the field of
quaternions $\mathbb{H}$. The specific construction contains naturally the
supertwistor one of the previous work by Litov and Pervushin \cite{litov:1984}
and it is shown that in the case of extended supersymmetry such an approach
leads to the separation of a class of superspaces and and its groups of
motion. We briefly discuss this particular extension to the domain of
quaternionic superspaces as nonlinear realization of some kind of the affine
and the superconformal groups with the final end to include also the
gravitational field\cite{Borisov:1974bn} (this last possibility to include
gravitation, can be realized on the basis of the
reference\cite{DiegoSalvatore} where the coset $\frac{Sp(8)}{SL(4R)}\sim$
$\frac{SU(2,2)}{SL(2C)}$was used in the non supersymmetric case). It is shown
that this quaternionic construction avoid some unconsistencies appearing at
the level of the generators of the superalgebras (for specific values of $p$
and $q$; $p+q=N$) in the twistor one.

\end{abstract}
\volumeyear{year}
\volumenumber{number}
\issuenumber{number}
\eid{identifier}
\date[Date text]{date}
\received[Received text]{date}

\revised[Revised text]{date}

\accepted[Accepted text]{date}

\published[Published text]{date}

\startpage{1}
\endpage{102}
\maketitle
\tableofcontents

\section{Introduction}

There are three main approaches to the construction of supersymmetric theories
\cite{volkov:1973}. In the first the supersymmetry is realized on fields,
while in the second directly on the superspace. The third approach is based on
the assumption of the simplest structure elements of superspace called
supertwistors \cite{litov:1984,ferber:1978}.

The particular geometrical environment of this approach, plus the explicit
covariant formulation appears as the main advantage over the standard
component one. The most important problem of the superfield approach is the
explicit formulation in terms of unconstrained superfields \cite{ferber:1978}.
One way to deal with this issue is with twistors. If one begins with twistors,
the compact complex version of the Minkowski superspace $M$ is naturally
realized as a flag space. The geometry of the flat case corresponding to $N=1$
(the minimal number of odd coordinates) turns into the geometry of the simple
supergravity model of Ogievetsky and Sokatchev \cite{og-sok:1978} after a
convenient twist.

In this paper, because we are interested in the the number of parameters of
the super-Lorentz transformations and the number of the Goldstone modes, we
extend the supertwistor construction of refs. \cite{litov:1984} to the
quaternionic one. We expect that this particular extension, that is justified
by the Ogievetsky theorem \cite{ogievetsky:1973}to the domain of superspaces,
will be bring us the correct number of fields of the standard model as the
simultaneous nonlinear realization of some kind of the affine and the
superconformal groups [6-9], with the final end to include also the
gravitational field\cite{Borisov:1974bn} (this last possibility to include
gravitation, can be realized on the basis of the
reference\cite{DiegoSalvatore} where the coset $\frac{Sp(8)}{SL(4R)}\sim$
$\frac{SU(2,2)}{SL(2C)}$was used in the non supersymmetric csase).
Consequently, we close this paper with a short discussion about the kinds of
possible supergroups able to support a twistor and a quaternionic structure.

\section{Quaternionic construction}

The link between a (super) twistor space $\mathbb{T}$ and the (super)
quaternionic one $\mathbb{H}$ is through the natural symplectic structure of
$\mathbb{H}$ (see in other contexts, for example \cite{LAT}). Because the
minimal quaternionic realization is in a $2\times2$ complex matrix structure
(e.g. $R\otimes SU(2)$), the supertwistor construction only can be implemented
for an even number of components in a matrix realization as we will show
below. . Now we will construct the quaternionic superextension analog to the
twistorial one in paper \cite{litov:1984}.

\subsection{Supertwistors}

In the twistor theory our starting point is a complex space $\mathbb{C}%
M\sim\mathbb{C}_{2,4}\left(  T\right)  .$ By using the correspondence with
null twistors we can successfully separate from it, the real Minkowski space
$M$ invariant with respect to the conformal group. The complex space CM\ in
its compactified form, also contains $S^{4}.$ For this space, the twistor
correspondence can be introduced as follows. Reality of the space $S^{4}$
doesn't follows, as in the Minkowski case, from the nullification of some kind
on twistors. It follows from the invariance under an antilinear mapping
$\rho:\mathbb{C}^{4}\rightarrow\mathbb{C}^{4},\rho^{2}=-1.$ Then, is clear
that this mapping represents the multiplication by the standard quaternionic
unit due $\mathbb{C}^{4}\sim\mathbb{H}$. Because the the space $S^{4}$ is not
invariant under $SU(2,2)$ we have to take the covering group of the
complexified $SU(2,2)$ group. The direct possibility is, for example,
$SL(4,\mathbb{C})$ which covers simultaneously $SO(6,\mathbb{C})$ and $S^{4}$
being consequently invariant with respect to $SL(2,\mathbb{H})$: its real form
which covers $SO(5,1).$

When we pass to supertwistor space there exists similar mapping in analogy
with the $\rho$ in the non-susy case but only for $N$ even. Then, if we assume
$N=2(p+q)$ the quaternionic structure can be introduced. This can be realized
taking into account that the quaternions can be written as $\mathbb{H}%
=\mathbb{C}\otimes\mathbb{C}$ e.g.: the first quaternionic unit in the field
of $\mathbb{H}$\ is identified with the imaginary unit in $\mathbb{C}$.
Consequently, \ any $q\in\mathbb{H}$ can be written as%
\begin{equation}
q=a+b\widehat{i}_{2},a,b\in\mathbb{C}\left(  \widehat{i}_{1}\right)  \tag{1}%
\end{equation}
where $\mathbb{C}\left(  \widehat{i}_{1}\right)  $ is the complex space with
the first quaternionic generator as imaginary unit.

To make explicit connection between the quaternionic structures in
supertwistor spaces, we only need to introduce $\ 2\times2$ matrices by each
element of the standard supertwistors operators. For instance, is clear that
under conjugation any supertwistor(standard notation, internal fermion indices
dropped) $\left(  \omega,\pi,\theta\right)  $ goes to $\left(  \overline
{\omega}=\epsilon\omega^{\ast},\overline{\pi}=\epsilon\pi^{\ast}%
,\overline{\theta}=\epsilon\theta^{\ast}\right)  $ where:%
\begin{equation}
\epsilon=\left(
\begin{array}
[c]{cc}
& 1\\
-1 &
\end{array}
\right)  . \tag{2}%
\end{equation}

\begin{remark}
In the pure twistor theory we need to use of the correspondence with null
twistors in order to separate from the twistor complex space $C_{2,4}(T)$ the
real Minkowski space $M$ invariant with respect to the conformal group.
Because in the quaternionic case the twistor constructions don't follows from
nullification of some kind of twistors (as in Minkowski space) but through the
invariance under an antilinear mapping $\rho:\mathbb{C}^{4}\rightarrow
\mathbb{C}^{4},\rho^{2}=-1,$ \ ($\mathbb{C}^{4}\sim\mathbb{H)}$, we have not
singularities of the super-generators of the theory. Also is possible to have
an alternative consistent quantum-field theoretical construction to the
ligh-cone one.
\end{remark}

\subsection{L and R subspaces}

Although in the simplest case in four spacetime dimensions, it is neccesary to
find the elements from $\ L$ and $R$ subspaces invariant under the $\rho$ map.
We are interested in the $\left(  2,0\right)  $ and $\left(  2,N\right)  $
subspaces in the field of the quaternions. As is the standard supertwistor
case, we can stablish an incidendence condition to determining $\left(
2,0\right)  $ subspace in full quaternionic form as follows:
\begin{equation}
q=wp,\text{ \ \ }s=\overline{\eta}p \tag{3}%
\end{equation}
where $q,p$ and $s$ are quaternions (or higher dimensional quaternionic
matrices) constructed as in the previous paragraph, and $w$ and $\overline
{\eta}$ are quaternions obtained via correspondence:
\begin{equation}
a+b\widehat{i}_{2}\leftrightarrow\left(
\begin{array}
[c]{cc}%
a & -b^{\ast}\\
b & a^{\ast}%
\end{array}
\right)  \tag{4}%
\end{equation}
As it is clear, we get the quaternionic superspace $\mathbb{H}^{\left(  1\mid
N\right)  }$ with projective coordinates on $P\left(  \mathbb{H}^{\left(
2\mid N\right)  }\right)  $defined as $qp^{-1}=w,$ \ \ $sp^{-1}=\overline
{\eta}$ For example, if the supergroup $SL\left(  2,\mathbb{H}\mid N\right)  $
acts on $\mathbb{H}^{\left(  2\mid N\right)  },$ knowing that $SO\left(
5,1\right)  \subset SL\left(  2,\mathbb{H}\mid N,\mathbb{H}\right)  ,$ then we
obtain the transformations of the quaternionic superspace including the
$S^{4}$( Euclidean) conformal transformations. Let us to consider the
complexification of the $SU\left(  2,2\mid2N\right)  $ supergroup, namely the
supergroup $SL\left(  2,\mathbb{H}\mid N\right)  \sim$ $SL\left(
4,\mathbb{C}\mid2N\right)  $and let us to take its real form which preserves
the reality of the map $\rho$ in the sense, for example, which maps the given
subspace $S^{4}$ onto itself. The infinitesimal transformation from can be
written as:%
\begin{equation}
\left(
\begin{array}
[c]{c}%
\delta\omega_{\alpha}\\
\delta\pi^{\alpha^{\prime}}\\
\delta\theta_{i}%
\end{array}
\right)  =\left(
\begin{array}
[c]{ccc}%
l_{\text{ }\beta}^{\alpha}+\frac{\left(  D+G\right)  }{2}\delta_{\text{ }%
\beta}^{\alpha} & a^{\alpha\beta^{\prime}} & \psi_{\text{ }j}^{\alpha}\\
b_{\alpha^{\prime}\beta} & -k_{\alpha^{\prime}}^{\text{ }\beta^{\prime}}%
-\frac{\left(  D-G\right)  }{2}\delta_{\alpha^{\prime}}^{\text{ }\beta
^{\prime}} & \varphi_{\alpha^{\prime}j}\\
\xi_{\text{ }\beta}^{i} & \chi^{i\beta^{\prime}} & S_{\text{ }j}^{i}+\frac
{2G}{N}\delta_{\text{ }j}^{i}%
\end{array}
\right)  \left(
\begin{array}
[c]{c}%
\omega^{\beta}\\
\pi_{\beta^{\prime}}\\
\theta_{j}%
\end{array}
\right)  \tag{5}%
\end{equation}
where $l,a,b,k$ are quaternion-valuated parameters. These parameters of
$SL\left(  2,\mathbb{H}\mid N\right)  $ are restricted by the requirement of
conservation of $\rho$-invariant subspaces. Finally, due to the
quaternion-twistor correspondence, these transformations can be related with
the quaternion-valuated spaces $E_{\text{ }R}^{N},E_{\text{ }L}^{N}$and
$E_{0}.$

\subsection{Quaternionic superspace}

We know that the equations relating the $B_{0}$ and $B_{1}$ parts of the
superspace in the case of supertwistors can be in terms of quaternions as follows.

i) The fundamental representation can be descomposed, in principle, as in the
case of [1]\ as
\begin{equation}
\mathcal{U}=t\cdot h \tag{6}%
\end{equation}

where h is an element of the maximal compact subgroup $S(U(2m)\times U(2n))$
and $t$ of the corresponding coset space $\frac{SU(2m;2n)}{S(U(2m)\times
U(2n))}$. Explicitly%
\begin{equation}
h=\exp\left[  i\left(
\begin{array}
[c]{cc}%
\chi & 0\\
0 & \varepsilon
\end{array}
\right)  \right]  =\left(
\begin{array}
[c]{cc}%
\mu & 0\\
0 & \upsilon
\end{array}
\right)  \tag{7}%
\end{equation}
and
\begin{equation}
t=\exp\left[  i\left(
\begin{array}
[c]{cc}%
0 & \nu\\
\overline{\nu} & 0
\end{array}
\right)  \right]  =\frac{1}{\sqrt{1-\overline{Z}Z}}\left(
\begin{array}
[c]{cc}%
\mathbb{I} & Z\\
\overline{Z} & \mathbb{I}%
\end{array}
\right)  \tag{8}%
\end{equation}
where%
\begin{equation}
Z_{\text{ }B}^{A}=\left[  \frac{Tanh\sqrt{\overline{\nu}\nu}}{\sqrt
{\overline{\nu}\nu}}\nu\right]  _{\text{ }B}^{A} \tag{9}%
\end{equation}
with $\nu_{\text{ }B}^{A}$ is a quaternionic matrix that can be generically
represented by a complex $(2(m+p)\times2(n+q))-$matrix $\nu^{i\alpha}\left(
i=1......2\left(  m+p\right)  ;\alpha=1......2(n+q\right)  $ ) with%
\begin{equation}
\nu^{i\alpha}=\mathcal{J}^{ij}\mathcal{J}^{\alpha\beta}\overline{\nu}_{j\beta}
\tag{10}%
\end{equation}
here $\mathcal{J}^{ij},\mathcal{J}^{\alpha\beta}$ are the $(2m+2p;2n+2q)$
matrices \ (obvious generalization of the standard $\epsilon^{\alpha\beta}$
that introduces the corresponding symplectic structure\cite{AO}) such that:%
\begin{equation}
\mathcal{J}=-\mathcal{J}^{T}=-\mathcal{J}^{\dagger}=-\mathcal{J}^{-1} \tag{11}%
\end{equation}
The convenient representation is the canonical symplectic matrix%
\begin{equation}
\mathcal{J}=\left(
\begin{array}
[c]{cccccccc}%
0 & 1 &  &  &  &  &  & \\
-1 & 0 &  &  &  &  &  & \\
&  & 0 & 1 &  &  &  & \\
&  & -1 & 0 &  &  &  & \\
&  &  &  & \cdot &  &  & \\
&  &  &  &  & \cdot &  & \\
&  &  &  &  &  & 0 & 1\\
&  &  &  &  &  & -1 & 0
\end{array}
\right)  \tag{12}%
\end{equation}
consequently $\nu^{i\alpha}$ has $4(m+p)(n+q)$ degrees of freedom.

ii) The metric constraint is invariant under $Usp\left(  m\mid p\right)
\times Usp\left(  n\mid q\right)  \approx SU\left(  2m\mid2p\right)  \times
SU\left(  2n\mid2q\right)  $ transformations%
\begin{equation}
\nu^{\prime}=g\cdot\nu\cdot h \tag{13}%
\end{equation}
where $g,h$ are unitary matrices of dimension $2(m+p),2(n+q)$ respectively.%
\begin{equation}
g\cdot\mathcal{J}_{m}\cdot g^{T}=\mathcal{J}_{m}\text{ \ \ \ ,\ \ \ \ \ \ }%
h\cdot\mathcal{J}_{n}\cdot h^{T}=\mathcal{J}_{n}; \tag{14}%
\end{equation}
This constraints and the unitarity condition are the defining propierties of
the $Usp-$groups, even in the SUSY\ case. The Lie algebra is spanned by
independent complex $2m\times2n-$dimensional generators $A=\mathcal{J}\cdot
X$, with%
\begin{equation}
A=-A^{\dagger}\text{ \ \ \ or }X=\mathcal{J}\cdot X^{\dagger}\cdot\mathcal{J}
\tag{15}%
\end{equation}
and%
\begin{equation}
A\cdot J+J\cdot A^{\dagger}=0\text{ \ \ or \ \ }\ X=X^{T} \tag{16}%
\end{equation}

\section{The super-Hilbert space}

The starting point is the Lie superalgebra coming from the Poisson structure
of the Manifold:%
\begin{equation}
\lbrack f,g]=f\frac{\overleftarrow{\partial}}{\partial\widetilde{T}^{R}%
}\left(  \omega^{-1}\right)  _{\text{ }S}^{R}\frac{\overrightarrow{\partial}%
}{\partial T_{S}}g \tag{17}%
\end{equation}
we obtain the following set of Poisson bracket relations between the the
quaternionic valuated variables:%
\begin{align}
&
\begin{array}
[c]{cc}%
\left[  \mu_{m}^{\alpha},\lambda_{\beta n}^{\text{ }}\right]  =\delta_{\text{
}\beta}^{\alpha}\delta_{mn}, & \left[  \overline{\lambda}_{\overset{\cdot
}{\beta}n},\overline{\mu}_{m}^{\overset{\cdot}{\alpha}}\right]
\end{array}
=\delta_{\overset{\cdot}{\beta}}^{\text{ }\overset{\cdot}{\alpha}\text{\ }%
}\delta_{nm}\tag{18}\\
&
\begin{array}
[c]{cc}%
\left[  -\xi_{ir},\xi_{s}^{j}\right]  =-\delta_{i}^{\text{ }j\text{\ }}%
\delta_{rs}, & \left[  -\overline{\eta}_{kt},\eta_{u}^{l}\right]
\end{array}
=-\delta_{k}^{\text{ }l\text{\ }}\delta_{tu} \tag{19}%
\end{align}

\bigskip Therefore, for generalized quaternion-valuated supertwistors
$T,\widetilde{T}$ we straighforwardly have:%
\begin{equation}
\left[  \widetilde{T}^{R},T_{S}\right\}  =\delta_{\text{ }S}^{R} \tag{20}%
\end{equation}
with:%
\begin{equation}
T_{R}=\left(
\begin{array}
[c]{c}%
\overline{\xi}_{A}\\
\eta_{M}%
\end{array}
\right)  \widetilde{T}^{R}=\left(  \xi^{A},-\overline{\eta}^{M}\right)
\tag{21}%
\end{equation}
such that%
\begin{equation}
\overline{\xi}_{A},\eta_{M}\in\mathbb{H}_{n} \tag{22}%
\end{equation}

\bigskip The superspace $Z$ has the following general form(we follow notation
from[1])
\begin{equation}
Z_{B}^{\dagger\text{ }A}=\overset{}{\overset{2n+2q}{\left.  \overbrace{\left(
\begin{array}
[c]{cc}%
X_{m}^{\text{ }a} & \theta_{m}^{\text{ }i}\\
\chi_{l}^{\text{ }a} & \lambda_{l}^{\text{ }i}%
\end{array}
\right)  }\right\}  }{\small 2m+2p}} \tag{23}%
\end{equation}

it represents the space of parameters of quaternionic supercoherent states,
depending on the structure of the supercoset space. The above quaternionic
supermatrix acts over the following quaternionic supervectors%
\begin{align}
\xi^{A}  &  =\left(
\begin{array}
[c]{cc}%
a^{c}, & -\xi^{i}%
\end{array}
\right)  ;\text{ }\overline{\xi}_{A}=\left.  \left(
\begin{array}
[c]{c}%
a_{c}^{\dagger}\\
\xi_{i}^{\dagger}%
\end{array}
\right)  \right\}  {\small 2n+2q,}\tag{23}\\
\overline{\eta}^{M}  &  =\overset{2m+2p}{\overbrace{\left(
\begin{array}
[c]{cc}%
b^{\dagger m}, & \eta^{\dagger l}%
\end{array}
\right)  }};\text{ }\eta_{M}=\left(
\begin{array}
[c]{c}%
b_{m}\\
\eta_{l}%
\end{array}
\right)  \tag{24}%
\end{align}
where we have defined%
\begin{align}
a_{c}^{\dagger}  &  =\frac{1}{\sqrt{2}}\left(  \lambda_{\alpha}+\overline{\mu
}^{\overset{\cdot}{\alpha}}\right) \tag{25}\\
b_{m}  &  =-\frac{1}{\sqrt{2}}\left(  \lambda_{\alpha}-\overline{\mu
}^{\overset{\cdot}{\alpha}}\right)  \tag{26}%
\end{align}
The explicit superfield coherent state reads as%
\begin{gather}
\left\vert \Phi_{B...}^{\text{ }A...}\left(  Z\right)  \right\rangle
=e^{\overline{\eta}^{M}Z_{M}^{\dagger\text{ }A}\overline{\xi}_{A}}\left\vert
\Phi_{B...}^{\text{ }A...}\right\rangle _{0}\tag{27}\\
=\exp\left[  b^{\dagger m}\left(  X_{m}^{\text{ }a}a_{a}^{\dagger}+\theta
_{m}^{\text{ }i}\xi_{i}^{\dagger}\right)  +\eta^{\dagger l}\left(  \chi
_{l}^{\text{ }a}a_{a}^{\dagger}+\lambda_{l}^{\text{ }i}\xi_{i}^{\dagger
}\right)  \right]  \left\vert \Phi_{B...}^{\text{ }A...}\right\rangle _{0}
\tag{28}%
\end{gather}
The Grassmann character of the matrix coefficients, restricts the number of
terms in (28) (see Appendix):%
\begin{equation}
\Phi_{B...}^{\text{ }A...}\left(  Z\right)  =\overset{2p}{\underset{n_{1}%
=0}{\sum}}\frac{1}{n_{1}!}\left(  b^{\dagger m}\theta_{m}^{\text{ }i}\xi
_{i}^{\dagger}\right)  ^{n_{1}}\overset{2q}{\underset{n_{2}=0}{\sum}}\frac
{1}{n_{2}!}\left(  \eta^{\dagger l}\chi_{l}^{\text{ }a}a_{a}^{\dagger}\right)
^{n_{2}}\overset{\min(p,q)}{\underset{n_{3}=0}{\sum}}\frac{1}{n_{3}!}\left(
\eta^{\dagger l}\lambda_{l}^{\text{ }i}\xi_{i}^{\dagger}\right)  ^{n_{3}%
}f_{B...}^{\text{ \ }A...}\left(  x\right)  \tag{29}%
\end{equation}
where:%
\begin{equation}
f_{B...}^{\text{ \ }A...}\left(  x\right)  =e^{b^{\dagger m}X_{m}^{\text{ }%
a}a_{a}^{\dagger}}\left\vert \Phi_{B...}^{\text{ }A...}\right\rangle _{0}
\tag{30}%
\end{equation}

\section{Quaternionic supercoherent states}

\subsection{Some examples}

i) The basis in the simplest cases is the superalgebra $U\left(
1,1\mid1,\mathbb{H}\right)  $ that contains as subgroups $U\left(
1,1;\mathbb{H}\right)  $ $\sim SO\left(  1,4\right)  $ and $U\left(
1,\mathbb{H}\right)  \sim SO\left(  2\right)  .$It is determined by
infinitesimal transformations preserving the scalar product $\overline{q}%
_{a}q_{a}-\overline{q}_{b}q_{b}+\overline{\eta}e_{2}\eta$ where $\eta$ is a
standard Grassmann quaternion and $e_{2}$ is the $B_{1}$ part of the supermetric.

ii) For example, a more complicated case is in \ the field of $\mathbb{H}$
with $p=2$ and $q=2,$ (that is the quaternionic analog of the $SU\left(
2,2\mid8\right)  $ with $N=8=p+q\rightarrow p=4,q=4).$ In this case we have
the following scalar quaternionic superfield:
\begin{equation}
\Phi\left(  Z\right)  =f\left(  x\right)  +\overline{\theta}_{m}^{\text{ }%
i}\overline{\chi}_{k}^{\text{ }a}f\left(  x\right)  _{\left(  ia\right)
}^{\left\{  mk\right\}  }+\overline{\theta}_{m}^{\text{ }i}\overline{\theta
}_{n}^{\text{ }j}\overline{\chi}_{k}^{\text{ }a}\overline{\chi}_{l}^{\text{
}b}f\left(  x\right)  _{\left(  ijab\right)  }^{\left\{  mnkl\right\}  }
\tag{31}%
\end{equation}
which is a supermultiplet with helicities ranging up to $\left\vert
s\right\vert =2$ and the multiplicities of the quaternionic $N=4$ (complex
$N=8$) Maxwell supermultiplet. There is one more condition to be fulfilled by
the quaternionic wave function:%

\begin{equation}
\frac{1}{2}\widetilde{T}T\left\vert \Phi_{B...}^{\text{ }A...}\left(
Z\right)  \right\rangle =0 \tag{32}%
\end{equation}
because, in this particular case, $\widetilde{T}T$ is a $U\left(  1,1\mid
p+q,\mathbb{H}\right)  $ invariant quantity, this condition transforms into:
\begin{equation}
\left[  \widehat{s}-\frac{1}{2}\left(  F_{\xi}+F_{\eta}\right)  +\frac{1}%
{4}\left(  p+q\right)  \right]  \left\vert \Phi_{B...}^{\text{ }A...}\left(
Z\right)  \right\rangle =0 \tag{33}%
\end{equation}

where $\widehat{s}$ is the $U\left(  1,1,\mathbb{H}\right)  $- invariant
helicity operator and $F_{\xi},F_{\eta}$ are the quaternion-valuated fermion
number operators. The above condition plus the annihilation of the vacuum by
all $L^{-}$ and $R^{-}$ operators determine the structure of the lowest states
univoquely as follows:%
\begin{equation}
\left\vert \Phi_{B...}^{\text{ }A...}\right\rangle _{0}=\exp\left[  b^{\dagger
m}\left(  X_{m}^{\text{ }a}a_{a}^{\dagger}+\theta_{m}^{\text{ }i}\xi
_{i}^{\dagger}\right)  +\eta^{\dagger l}\left(  \chi_{l}^{\text{ }a}%
a_{a}^{\dagger}+\lambda_{l}^{\text{ }i}\xi_{i}^{\dagger}\right)  \right]
a_{\left\{  a\right.  }^{\dagger}a_{\left.  b\right\}  }^{\dagger}%
b^{\dagger\left\{  m\right.  }b^{\left.  \dagger n\right\}  }\left\vert
0,0\right\rangle \tag{34}%
\end{equation}

\begin{remark}
Notice that the quaternion-Casimir given by expression (33) shows that for
$p=q$ in this $\mathbb{H}$-formulation we have fundamental representation in a
sharp contrast with the purely $\mathbb{C}$-supertwistor one (where the last
term into the helicity operator goes as $\left(  p-q\right)  $ )(see [1]).
\end{remark}

\section{Concluding remarks and outlook}

The results concerning to this preliminary work can be enumerated in the
following points:

1) We have been extended the supertwistor construction \cite{litov:1984} to
the quaternionic one.

2) We are capable to extend, accordling to this new quaternionic description,
the number of fermionic fields beyond the pure supertwistorial one avoiding
(due the division ring structure) the singularities in the representation of
the supergenerators.

3) Some corresponding cosets (with some examples) have been identified,
remaining they as a characteristic subset of the quaternionic superextensions
with the full supertwistor propierties (e.g. arising from the super light cone structure).

4) We have obtained the corresponding coherent super-quaternionic states
spanning the super-Hilbert spaces.

In a separate paper \cite{DV}, the interesting cosets from which we are able
to perform the nonlinear realization in order to obtain the super-analog of
the Borisov-Ogievetsky one \cite{Borisov:1974bn} (e.g. to obtain the
corresponding number of Goldstone fields for the standard model) will be
discussed and an alternative to the light cone construction will be perfomed.
Moreover, the more important task that remains is to perform explicitly the
super-analog of the Borisov-Ogievetsky nonlinear realization \ in the same way
as \cite{DiegoSalvatore} developing consequently the same analysis and
physical construction as\cite{AndrejVict}

\section*{Acknowledgments}

The authors would like to thank A.B.~Arbuzov, A.A.~Zheltukhin, and A.E.~Pavlov
for useful discussions. D.J.C-L. is grateful to the JINR Directorate for
hospitality and to CONICET-Argentina\ for financial support.This work is in
memory of the professor and friend Boris Moiseevich Zupnik, one of the main
researchers in the supersymmetry, supergravity and other areas of the modern
mathematical physics, that pass away recently.

\section{Appendix: Grassmann quaternion}

\begin{theorem}
A Grassmann quaternion, as in the case of the standard one with coefficients
$\in\mathbb{R}$ can be written as $\Psi=A\left(  r\right)  e^{i\theta\Sigma}$
with $A$ and $\Sigma$ matrices depending of four different (anticommuting)
Grassmann numbers, with $\Sigma^{2}=1.$
\end{theorem}

\begin{proof}
a Grassmann quaternion is described by the following matrix
\end{proof}

\begin{equation}
\Psi=\left(
\begin{array}
[c]{cc}%
\psi_{0}+i\psi_{3} & -\psi_{2}+i\psi_{1}\\
\psi_{2}+i\psi_{1} & \psi_{0}-i\psi_{3}%
\end{array}
\right)  \tag{35}%
\end{equation}
with $\psi_{a}$ $(a=0...3)$Grassmann numbers. We introduce polar coordinates
as%
\begin{equation}
\left.
\begin{array}
[c]{c}%
\psi_{0}=r\cos\theta=r\cdot1\\
\psi_{1}=r\sin\theta\sin\phi\cos\chi=r\cdot\theta\cdot\phi\cdot1\\
\psi_{2}=r\sin\theta\sin\phi\sin\chi=r\cdot\theta\cdot\phi\cdot\chi\\
\psi_{3}=r\sin\theta\cos\phi=r\cdot\theta\cdot1
\end{array}
\right\}
\begin{array}
[c]{c}%
Grassmann\\
coefficients
\end{array}
\tag{36}%
\end{equation}
where $r,\theta,\phi,\chi$ are also Grasmann numbers. Then (35)\ can be
written as follows:%
\begin{equation}
\Psi=A\left(  r\right)  e^{i\theta\Sigma} \tag{37}%
\end{equation}
with%
\begin{equation}
A\left(  r\right)  =r\sigma_{0}\text{ \ ,\ \ \ \ \ }\sigma_{0}=\mathbb{I}%
_{2\times2}\text{\ } \tag{38}%
\end{equation}%
\begin{equation}
\Sigma=\left(
\begin{array}
[c]{cc}%
\cos\phi & \sin\phi e^{i\chi}\\
\sin\phi e^{-i\chi} & -\cos\phi
\end{array}
\right)  =\left(
\begin{array}
[c]{cc}%
1 & \phi\left(  1+i\chi\right) \\
\phi\left(  1-i\chi\right)  & -1
\end{array}
\right)  \tag{39}%
\end{equation}
where the last matrix in the $RHS$ of the above expression coming from the
Grassmann properties of the corresponding coefficients. Notice that
automatically $\Sigma^{2}=1$ , consequently the concrete construction proposed
in this paper is faithful and consistent with the $\mathcal{J}-matrix$ and the
quaternionic supervectors $\xi^{A}$ and $\eta_{M}.$



\end{document}